\documentclass{article}
\usepackage{times}
\usepackage{amsmath,amsfonts,amsthm,amssymb,amsbsy,amstext,amscd,amsxtra,multicol}

\newtheorem{corollary}{Corollary}
\newtheorem*{corollary*}{Corollary}
\newtheorem{theorem}{Theorem}
\newtheorem*{theorem*}{Theorem}
\newtheorem{lemma}{Lemma}
\newtheorem*{lemma*}{Lemma}
\newtheorem{prop}{Proposition}

\theoremstyle{definition}
\newtheorem*{examplesm}{Examples}

\newtheorem{defi}{Definition}
\newtheorem{rem}{Remark}

\title{On complexity of regular realizability
  problems\thanks{Supported by RFBR grants 
    11--01--00398 and 12--01--00864}}
\author{Mikhail N. Vyalyi\\ \texttt{vyalyi@gmail.com}}
\date{December 29, 2012}

\let\eps\varepsilon
\let\al\alpha
\let\leq\leqslant			    		 
\let\geq\geqslant			    		 

\let\vk\varkappa
\let\ph\varphi

\def\NN{\mathbb N}
\def\es{\varnothing}
\let\sm\setminus
\def\bydef{\mathrel{\stackrel{\scriptscriptstyle\text{def}}{=}}}
\newcommand*\RE{\ensuremath{\mathrm {\Sigma_1}}}
\newcommand*\DL{\ensuremath{\mathrm {LOG}}}
\newcommand*\NL{\ensuremath{\mathrm {NLOG}}}
\newcommand*\Pclass{\ensuremath{\mathrm {P}}}
\newcommand*\NP{\ensuremath{\mathrm {NP}}}
\newcommand*\BPP{\ensuremath{\mathrm {BPP}}}
\newcommand*\EXP{\ensuremath{\mathrm {EXP}}}
\newcommand*\FNL{\ensuremath{\mathrm {FNL}}}

\newcommand*\poly{\ensuremath{\mathrm {poly}}}
\newcommand*\PSPACE{\ensuremath{\mathrm {PSPACE}}}
\newcommand*\DSPACE{\ensuremath{\mathrm {DSPACE}}}

\newcommand*\NSPACE{\ensuremath{\mathrm {NSPACE}}}

\def\Im{\mathop{\mathrm{Im}}}

\def\ex{\ensuremath{\mathrm {sq}}}
\def\dex{\ensuremath{{\mathrm {sq}^\uparrow}}}
\def\D{D^{\uparrow}}

\def\C{\ensuremath{\mathcal C}}

\def\co{\text{\textup{co-}}}

\def\Iff{\mathbin{\;\Leftrightarrow\;}}

\def\Seq{\mathrm{Seq}}
\def\Iseq{\mathrm{ISeq}}
\def\reg{\mathrm{RR}}
\def\leGen#1#2{\mathbin{\leq^{\mathrm{#2}}_{\mathrm{#1}}}}

\def\lelog{\leGen{m}{log}}
\def\lenl{\leGen{m}{FNL}}

\def\ledtnl{\leGen{dtt}{FNL}}
\def\ledtl{\leGen{dtt}{log}}
\def\ledtp{\leGen{dtt}{p}}

\let\la\langle
\let\ra\rangle

\begin{document}
\maketitle
\def\abstractname{}
\def\proofname{Proof}
\def\refname{References}

\begin{abstract}
  A regular realizability (RR) problem is testing nonemptiness of
  intersection of some fixed language (filter) with a regular
  language. We show that RR problems are universal in the following
  sense. For any language $L$ there exists RR problem equivalent
  to~$L$ under disjunctive reductions on nondeterministic log space.
  
  We conclude from  this result an existence of complete problems under
  polynomial reductions for many  complexity classes including all
  classes of the polynomial hierarchy.
\end{abstract}

Motivation of this work is to find out a specific class of algorithmic
problems that represents in a unified way complexities of all known
complexity classes (there are hundreds of them now).

A typical algorithmic problem is recognition problem for a
language. But in the most interesting cases an input is structured:
it is a graph, function description etc. Our main goal is to choose a
structure of an input to satisfy two (somewhat contradictory)
requirements: a specific class of problems should be wide enough and
it should be useful. The latter requirement reflects a hope that
analysis of a specific problem might be easier than a general case.

Here we consider regular realizability problems in this context.
Let  $L$ be a language (we refer it further as \emph{a filter\/}). \emph{The
regular realizability} problem with this filter is a question about
realizability of regular properties on words in~$L$. More exactly, the
language  $\reg(L)$ consists of descriptions of regular languages  $R$
such that  $R\cap L\ne\es$. 

What are possible complexities of RR problems? In this paper we
obtain a~partial answer on this question. It appears
that RR problems are universal: for any other 
problem there exists an equivalent RR problem.

To make exact statements we need to fix a format for
descriptions of regular languages and an equivalence relation. 

We represent a regular language $R$ by a deterministic finite
automaton (DFA) $A$ recognizing the language $R$ (and denote this fact
as $R=L(A)$). DFAs are described in natural way by transition tables.
Details of the format used can be found in~\cite{Vya11}.  Important
features for this work are: (i) each binary word $w$ is a description
of some DFA $A(w)$ and (ii) testing membership for a regular language
$L(w)=L(A(w))$ can be done using deterministic log space. So formal
definition of the language $\reg(L)$ corresponding to RR problem with
a filter~$L$ is
$$
w\in\reg(L)\;\Leftrightarrow\; L\cap L(A(w))\ne \es.
$$

Equivalence relations considered here are induced by algorithmic
reductions. For any language $X$ an RR-representative of~$X$ (under
reductions of some type) is a
language $L$ such that  $\reg(L)$ is equivalent to  $X $ under
reductions of this type.

The most natural reductions in the context of regular realizability
are $m$-reductions using log space ($\lelog$-reductions). There exist
RR problems complete under log space reductions for complexity classes
such as $\DL$, $\NL$, $\Pclass$, $\NP$, $\PSPACE$, $\EXP$, $\RE$,
see~\cite{ALRSS09,Vya09}.

We can prove universality result mentioned above for stronger
reductions: disjunctive reductions using nondeterministic log space
(definition see below in Section~\ref{red-sect}).  In the proof we
reduce an arbitrary language $X$ to RR language by a
\emph{monoreduction\/}: a reduction of special kind that sends a
word~$w$ to a description of DFA accepting 1-element set $\{f(w)\}$,
where $f$ is an injective map.  Disjunctive nlog space reductions
appear in an attempt to invert a monoreduction (see
Section~\ref{univ-sect}).

For many cases disjunctive nlog space reductions are weaker than
polynomial reductions. In this way we extend a list of classes having
complete RR problems under polynomial reductions. It will be shown in
Section~\ref{normal-sect} that there exist RR complete problems for
each class of the polynomial hierarchy. Note that it is a rather
surprising even for the class $\co\NP$: RR problem is formulated by
existential quantifier (an existence of an accepting path possessing
specific properties) and there is no direct way to convert it into
universal quantifier.

We also present two other universality results.  

In Section~\ref{mono-sect} we show that RR problems are universal for
promise problems. In this case inversion of a monoreduction is much
simpler. But it requires \FNL{} reductions too. It seems that using
nondeterministic log space is unavoidable (see Remark~\ref{Nlogness}
in Section~\ref{mono-sect}).

In Section~\ref{GNA-sect} we extend universality to the classes of
generalized nondeterminism introduced in~\cite{Vya09}. 
Technically it suffices to prove universality of RR problems with
prefix closed filters. 
The proofs in
this section are similar to the proofs in Section~\ref{univ-sect} but
they involve more technical details.

Thus RR problems are universal. They represent a huge variety of
complexities. Are they useful? The proof suggests a negative
answer. Reductions used in the proof cut off almost all properties of regular
languages and put `the hard part' of a problem into instances
corresponding to finite languages. Of course, nlog reductions say
nothing about languages inside $\NL$. 




\section{What reductions are used in universality results}\label{red-sect}

Equivalence relations used here are defined using reductions. Languages
are equivalent (notation $A\sim B$) if they are reduced to each other under
reductions of some type.

We recall some basic definitions concerning algorithmic reductions. 
Let  $\C$ be a function class.  A language $A$ is reduced to a
language $B$ under $m$-reductions by functions from the class (notation
$\leGen{m}{\C}$) if there exists a function $f\in\C$ such that  $x\in
A$ iff $f(x)\in B$. If the class $\C$ contains the identity map and
is closed under compositions then $m$-reduction relation is a
preorder, i.e. transitive and reflexive relation. 

The most known reductions of this type are polynomial reductions (see,
e.g.~\cite{AB}).

In this paper we need weaker $m$-reductions. The corresponding
functions are computed by deterministic Turing machines using space
logarithmically bounded w.r.t. the input length. We denote these reductions by~$\lelog$.  More exactly, the
machine has the read only input tape, the work tape of size $O(\log
n)$, where $n$ is an input size and write only output tape.
It easy to check that the relation $\lelog$ is transitive. Note that
the size of $f(x)$ is polynomially bounded by the size of~$x$ (the
number of configurations is polynomially bounded). So, the log space
algorithm for a composition  $f(g(x))$ uses a subroutine to compute
$i$th bit of  $g(x)$ to simulate the input $g(x)$ for the algorithm
computing~$f$. More details can be found in textbooks on complexity
theory, say,~\cite{AB, Sipser}.

The second important type of reductions is Turing
reductions. A~language $A$ is Turing reducible to a language $B$ if
there exists an algorithm recognizing~$A$ that uses~an oracle~$B$.
There are several restricted forms of Turing reductions. In
\emph{truth-table reductions} a reducing algorithm generates a list
of oracle queries $q_1,\dots, q_s$ and a Boolean function~$\al$
depending on $s$ arguments. Then algorithm asks all queries and
outputs $\al(\chi_B(q_1),\dots,\chi_B(q_s))$.

Note that for log space Turing reductions and truth-table
reductions are equivalent~\cite{LL76}.

We use in the main result a weaker form of truth-table 
reductions---disjunctive reductions (notation $\leGen{dtt}{}$). 
In this case the function $\al$ is disjunction. 

Disjunctive reductions can be expressed via $m$-reductions in the
following way. Define a language   $\Seq(X)$ as a collection of words
in the form
$$
\# x_1\# x_2\#\dots x_n\#,
$$
where $x_i\in X$ for some~$i$,  $\#$ is a delimiter (an additional
symbol that do not belong to the alphabet of the language~$X$). 
The following statement is clear from definition. 

\begin{prop}\label{dtt<->Seq}
  $A\leGen{dtt}{} B$ iff
  $A\leGen{m}{}\Seq(B)$.
\end{prop}

Words from $\Seq(X)$ can be identified with finite
sequences of words from~$X$ in natural way. Hereinafter we assume this
correspondence. 

\begin{defi}
  A class of languages is \emph{normal\/} if it is closed under the map
  $X\mapsto \Seq(X)$ and 
  $\lelog$-reductions: if $X\in \C$ and $Y\lelog X$ then
  $\Seq(X)\in\C$ and $Y\in\C$.. 
\end{defi}

\begin{defi}
  A language $X$ is $\leGen{m}{}$-\emph{normal\/} if $X\sim \Seq(X)$
  under $\leGen{m}{}$-reductions. 

  If a function class is not indicated we assume
  log space 
  reductions. 
\end{defi}

\begin{prop}\label{SeqIsNormal}
  $\Seq(X)$ is normal for any $X$. 
\end{prop}
\begin{proof}
  It is obvious that $Y\lelog \Seq(Y)$ for any~$Y$. Thus we need to
 prove that $\Seq(\Seq(X))\lelog \Seq(X)$. Note that 
$$\Seq(\Seq(X))\leGen{dtt}{log} X.$$ 
Indeed, a disjunction of disjunctions is a disjunction. So the
reducing algorithm form a query list consisting of all elements of all
elements of an input sequence.

To complete
 a proof apply Proposition~\ref{dtt<->Seq}.
\end{proof}

\begin{lemma}\label{normal-languages}
  Let $\Seq (X)\lelog Y$, $Y\lelog \Seq(X)$. Then $Y$ is normal.
\end{lemma}
\begin{proof}
  We denote reducing maps by  $f\colon
  (x_1,\dots,x_t)\mapsto y$ (for the first reduction) and
  $g\colon y\mapsto(x_1,\dots,x_t) $ (for the second one).

  Now we reduce  $\Seq(Y)$ to $Y$. Compute
  sequences~$g(y_i)$ for 
  all elements
  $y_i$ of the sequence  $(y_1,\dots,y_s)$ and merge all of them into
  a sequence  
  $h(y)$ consisting of words from the language~$X$. 
  Apply to this sequence the map~$f$. It is a reduction in question. 
  Correctness stems from the fact that 
  $y\in\Seq(Y)$ iff $h(y)\in \Seq(X)$ iff
  $f(h(y))\in Y$.
\end{proof}

\begin{corollary}\label{normal-classes}
  If a normal class contains $m$-complete languages then all complete
  languages in this class are normal.
\end{corollary}

Now we discuss a choice of a function class for disjunctive reductions. 

It is natural that  we prefer the weakest possible class. 
Log space seems to be unsufficient (see Remark~\ref{Nlogness} below). 

Next step is to use nondeterministic space. The corresponding class is
denoted by \FNL.  
Machines computing $\FNL$ functions use a logarithmically bounded work
tape and an unbounded oracle tape which is one way and write only.
An oracle puts its answer on the oracle tape and overwrites a
query. More details on the class~\FNL{} and its companions can be
found in~\cite{ABJ95}. In particular, the composition lemma holds for the
class~\FNL{} and the size of output is polynomially bounded by the input
size. 

We denote $\FNL$ disjunctive reductions by~$\ledtnl$.  The $\FNL$
$m$-reduction is denoted by $\lenl$. 

It is clear from definition that $\ledtnl$-reductions are stronger
then log space reductions and are weaker then disjunctive
reductions in polynomial time. So, we have the following facts.

\begin{prop}
  If $A\ledtl B$ then $A\ledtnl B$.
\end{prop}

\begin{prop}
  If $A\ledtnl B$ then $A\ledtp B$. Moreover, if  $B$ is normal then 
  $A\lenl B$. 
\end{prop}

\begin{rem}
  The universality is much easier to prove for exponential time
  reductions. But these reductions do not say anything about the most
  interesting realm of comple\-xi\-ties--- \PSPACE{} and below. 

  For polynomial time reductions we need the same constructions that
  are used in our proof below. The algorithmic facts become
  easier. 

  We use \FNL{} reductions to cover \Pclass{} and below.  Note that as
  shown in~\cite{AO96} basic counting log space classes are closed
  under $\ledtnl$-reductions.  
\end{rem}

\section{Examples of normal classes}\label{normal-sect}

We will apply Corollary~\ref{normal-classes} and
$\ledtnl$-universality results to prove that a complexity class
contains complete RR problems. Corollary~\ref{normal-classes} requires
a class into consideration to be normal.

The most of known complexity classes are normal. We give several examples.
In definitions of complexity classes and computational
models we follow Arora and
Barak book~\cite{AB}.

A straightforward algorithm for recognizing language $\Seq(X)$ is to check
$x_i\in X$ for all $x_i$ from the input 
$$\# x_1\#\dots \# x_m\#$$ and take disjunction of the results.

Let $X\in\DSPACE(f(n))$, where $f(n)=\Omega(\log n)$. Then the above
algorithm uses space $O(f(n))$. Thus the class $\DSPACE(f(n))$ is
normal (it closed under $\lelog$-reductions by obvious reasons). 

With small modifications the same argument is applied to
nondeterministic space classes. Nondeterministic algorithm guesses\footnote{Hereinafter a
  guess is a nondeterministic choice.} $i$
such that $x_i\in X$ and runs an algorithm recognizing~$X$ on an
instance~$x_i$. So the classes $\NSPACE(f(n))$, $f(n)=\Omega(\log n)$,
are also normal.  

For time complexity classes closeness under $\lelog$-reductions holds
for the class $\Pclass$ of polynomial time and more powerful
classes. Running time of the above algorithm recognizing $\Seq(X)$ is
upperbounded by
\begin{equation}\label{superadd}
\tilde T(n) =n+\max_{n=n_1+\dots+n_m} (T(n_1)+\dots+T(n_m)),
\end{equation}
where $T(n)$ is a running time of an algorithm recognizing~$X$.  It is
clear that $T(n)=\poly(n)$ implies $\tilde T(n) = \poly(n)$. So
$\Pclass$ is normal. For more powerful classes normality of a time
complexity class follows from closeness of time limitations under the
map $T(n)\mapsto \tilde T(n)$.  There is a~simple sufficient condition
for closeness: if a function $T(n)$ satisfies time limitations then
$nT(n)$ is also satisfies time limitations. Applying this observation
we get normality for classes of quasipolynomial time, exponential
time, simple exponential time and similar limitations.

It is easy to see that under the same conditions nondeterministic time
classes are also normal.  

The last example are classes of the polynomial hierarchy. 

\begin{prop}\label{HierNormal}
Each class of the polynomial hierarchy is normal.
\end{prop}
\begin{proof}
Closeness 
under $\lelog$-reductions is clear for each class of polynomial
hierarchy. So it remains to show 
closeness under the map $X\mapsto \Seq(X)$.

Let $\Iseq(X)$ be a language consisting of $\#\la i\ra\# x_1\dots \#
x_m\# $ such that $x_i\in X$. Here the delimiter symbol $\#$ does not
belong to the alphabet of the language~$X$ and $\la i\ra$ denotes
binary representation of integer~$i$.  It is clear that
$X\sim^{\mathrm{log}}_{\mathrm{m}} \Iseq(X)$. Indeed, using log space
one can extract $i$th list element.

By definition of $\Iseq(X)$ we have
\begin{equation}\label{logbounded}
(w\in\Seq(X)) \quad\Iff\quad \exists i (\# i w \in \Iseq(X)).
\end{equation}
Note that
if $X\in\Sigma^p_k$ then
$\{\# i w : \# i w\in \Iseq(X)\}$ is in $\Sigma^p_k$.
Thus   $\Seq(X)$ is in $ \Sigma^p_k$ due to~\eqref{logbounded}. 

Suppose now that $X\in \Pi^p_k$. In this case $\{\# i w : \# i w\in
\Iseq(X)\}$ is in $\Pi^p_k$ and we have for some $V\in \Sigma^p_{k-1}$ 
and polynomial $p(\cdot)$ 
$$
(\# i w\in \Iseq(X)) \quad\Iff\quad \forall y 
(|y|\leq p(|w|))\land (\# iw\# y\in V).
$$
So we need to interchange
quantifiers in~\eqref{logbounded}. 
Due to this fact one can apply an equality
$$
\exists i\in S \forall y W(x,i,y)  = \mathop{\forall}\limits_{i\in S} 
y^{i} \bigvee_{i\in
  S}  W(x,i,y^{i}) ,
$$
where $y^i$ are new variables indexed by~$i\in S$. Thus $\Seq(X)\in\Pi^p_k $. 
\end{proof}

\section{Monoreductions}\label{mono-sect}

Without loss of generality we consider languages in binary alphabet. 

Let $f$ be an injective map $\{0,1\}^*\to\{0,1\}^*$.

A~\emph{monoreduction} is a map  
\begin{equation}\label{monored}
x\mapsto A_{f(x)},
\end{equation}
where $A_{w}$ is a description of minimal DFA\footnote{We assume
that construction of minimal automaton is fixed. For minimization
algorithms see textbooks on formal languages, e.g.~\cite{Ko97}.}
recognizing the 1-element language~$\{w\}$.  Note that the number
of states in minimal DFA coincides with the number of Myhill -- Nerode
classes~\cite{Ko97}. It is easy to verify that the number of Myhill --
Nerode classes for the language $\{w\}$ is just $|w|+2$ (all prefixes
of the word~$w$ plus one) and the function $w\mapsto A_w$ is computed
on deterministic log space.

Informally speaking, the map~\eqref{monored} assigns `names' for all
binary words in the form of automaton description. Injectivity
condition implies that names of different words are different. 

If a map $f(x)$ is log space computed then the
corresponding monoreduction is also log space computed. It reduces  
a nonempty language $\es\subset X\subseteq\{0,1\}^*$ to some RR
problem
\begin{equation}\label{genred}
  X\lelog \reg (Y).
\end{equation}
It is easy to see from definitions that~\eqref{genred} holds iff
\begin{equation}\label{monocond}
Y\cap f(\bar X)=\es,\quad 
Y\supseteq f(X). 
\end{equation}
In other words, $Y$ separates images of $X$ and  $\bar
X$ and $Y$ contains the image of~$X$.

Complexity of a reduction in the opposite direction depends heavily on~$Y$.

Take for example $Y=f(X)$. It is not a good choice because complexity
of the language $\reg(f(X))$ varies in wide range w.r.t. complexity of
the language~$X$.  It can be illustrated in the simplest case
$f=\mathrm{id}$.

There exists a filter $L$ such that (a) the membership problem for $L$
is in the class of languages recognized by RAM in linear time; (b) 
$\reg(L)$ is complete for the class $\Sigma_1$
of recursively enumerable languages
under  $m$-reductions~\cite{Vya09}.

On the other hand for a  $X=\{0^n\}$ a language $\reg(X)$ is in
$\DL$. It was shown in~\cite{Vya11} that $\DL$ is $\lelog$-reducible
to any RR problem with infinite filter. 

Nevertheless, the reduction $X\lelog\reg(X)$ can be inverted if we
consider reductions among promise problems. A~\emph{promise problem}
is a problem of computing a partially defined predicate. In other
words there are two languages $L_1$ and $L_0$ such that $L_1\cap
L_0=\es$. The question is to test membership $w\in L_1$ provided
either $w\in L_1$ or $w\in L_0$.

Promise problems have more expressive power than languages
(which correspond to total  predicates). For many complexity
classes, say $\NP\cap\co\NP$ or $\BPP$, an existence of complete
languages in a class is an open problem. But there are simple and
natural examples of complete promise problems for these classes.

The question about complete RR  promise problems is also much easier
than the question about complete RR languages.

\begin{theorem}\label{promise}
  Any promise problem $(L_1,L_0)$ with $L_1\ne\es$ is equivalent to RR promise
  problem under nlog space reductions.
\end{theorem}

\begin{proof}
  Define $\reg(L_1: |R|=1)$ as RR problem with a promise $|L(A)|=1$,
  i.e. a pair of languages
  $$
  (\{ A: L(A)\cap L_1\ne\es\land|L(A)|=1\},
  \{ A: L(A)\cap L_1=\es\land|L(A)|=1\}).
  $$
  Then a promise problem
  $(L_1,L_0)$ is $\lelog$-reduced to the  problem
  $\reg(L_1: |R|=1)$ by the map $w\mapsto A_w$. 
 
  In the other direction we can prove a weaker reduction
  \begin{equation}\label{fromPromiseRR}
    \reg(L_1: |R|=1)\lenl (L_1,L_0).
  \end{equation}
  To construct a reduction~\eqref{fromPromiseRR} we need a procedure
  that find a~word accepted by DFA $A$ provided $A$ recognizes a
  1-element language.  This procedure is easily implemented in the
  class $\FNL$: it nondeterministically guesses  the word symbol by symbol
  maintaining the current state of the automaton reading the word.
\end{proof}

\begin{rem}\label{1-point}
We call an automaton \emph{simple} if it recognizes a 1-element set.
Testing simplicity can be done using nlog space.

Testing algorithm again guesses a word accepted by the
automaton symbol by symbol. But now it also checks a possibility to choose
more than one symbol in such a way that after each choice the
automaton can be moved to an accepting state by reading some sequence
of symbols.  Here we use a well-known equality  $\co\NL=\NL$
\cite{Im88,Sipser}. 
\end{rem}

\begin{rem}\label{Nlogness}
  Is it necessary to use $\NL$-oracle in the
  reduction~\eqref{fromPromiseRR}? The question is open but the
  negative answer is more plausible. In the non-uniform settings the
  class of unambiguous nondeterministic log space coincides with the
  class of nondeterministic log space~\cite{RA02}. It is
  quite natural to suggest that  simplicity
  test do not belong to $\DL$ if $\DL\ne\NL$.
\end{rem}

\begin{rem}\label{NotAuto}
  The unique  word accepted by DFA can be easily recovered from special forms
  of DFA description (say, description of minimal DFA accepting
  1-element language). But it does not help to improve the
  reduction~\eqref{fromPromiseRR} because we are interested in regular
  realizability problems (the answer depends on a language and should
  be the same for all automata recognizing the language).
\end{rem}

\section{Universality of RR problems}\label{univ-sect}

In the case of language reductions we are able to prove a weaker
version of Theorem~\ref{promise} using disjunctive reductions.

\begin{theorem}\label{first}
  For any non-empty language $X$ there exists a filter~$L$ such that
  \begin{equation}
    X\lelog \reg(L)\ledtnl X.
  \end{equation}
\end{theorem}

In the proof of Theorem~\ref{first} we use the other extreme case for
conditions~\eqref{monocond}. Namely, we choose
$Y=X_f\bydef\overline{f(\bar X)}$. 

The idea behind the proof is to approximate
inversion of a monoreduction as close as possible. The difficulty of inversion
stems from the fact mentioned in Remark~\ref{NotAuto}: instances of RR
problem are all regular languages and  RR problem might be
hard for languages that are not 1-element languages constituting the image of
a monoreduction. To overcome this difficulty we
choose a filter satisfying conditions~\eqref{monocond} in such a way
that for most regular languages RR problem  with the chosen filter is
trivial. 

The first step toward implementation of this idea is to make RR
problem trivial for all infinite languages.

\begin{defi}
  An infinite language is \emph{regularly immune\/} if it does not
  contain any infinite regular language.
\end{defi}

Suppose that $f(\{0,1\}^*)$ is contained in some regularly immune
language.  Then any infinite regular language has a non-empty
intersection with $\overline{f(\{0,1\}^*)}$. So for any infinite
instance of $\reg(X_f)$ the answer is positive.

For finite instance  of $\reg(X_f)$ $m$-reducing algorithm should indicate
a word that possibly belongs to~$X$. It seems very hard for
arbitrary~$X$. 

Therefore the second step is to use disjunctive reductions. In this
case reducing algorithm should just produce list of all words accepted
by a specific automaton and this task is much easier in general case.

Note that examples of regularly immune languages are numerous and easy
to construct. Any such language can be used in reduction outlined
above. 
But the cardinality of a finite language can be exponentially
larger than the number of states in DFA recognizing the language. So
such a reduction requires an exponential time. 

Thus the next step is to choose a regularly immune set~$D$ possessing
a special property: the cardinality of any regular language in $D$ is
polynomially upperbounded by the number of states in DFA recognizing
the language. To guarantee a polynomial upper bound it is sufficient
to require that $D$ consists of squares only (see
Remark~\ref{poly-bound}). Below in Lemma~\ref{Dimmune-few} we give a
linear bound.

Last but not least, all actions mentioned above should be efficiently
implemented. We are going to construct a $\ledtnl$-reduction. So 
we need $\FNL$ implementations.

To realize the plan outlined above we start from specification
of~$D$. Let $\beta$ be a morphism $\beta(0)=01$, $\beta(1)=10$ and
$\ex(\cdot)$ be a map
\begin{equation}\label{ex-def}
\ex\colon x\mapsto \beta(x)1^20^{|x|^2+3} \beta(x)1^20^{|x|^2+3}.
\end{equation}
We choose $D=\Im(\ex)$, i.e. an image of all binary words under
the map~$\ex$.

To prove that $D$ is  regularly  immune one can use Parikh
theorem~\cite{Ko97,Parikh}. It says that the lengths of words from a
regular (or even CFL) language form a \emph{semilinear set\/}, i.e. a
finite union of arithmetic progressions.

\begin{lemma}\label{Dimmune}
  $D$ is regularly immune.
\end{lemma}
\begin{proof}
  Lengths of words from $D$ form the set
  $\{2n^2+4n+10: n\in\NN\}$. It is clear that intersection of this set
  with any arithmetic
  progression is finite.
\end{proof}

To give an upper bound on the cardinality of a regular language
contained in~$D$ we make a couple of observations.

By definition,  $D$ words are squares. 
Moreover, they are incomparable in the
following sense. 

\begin{prop}\label{common-prefix}
  Let $pq_1, pq_2\in D$. Then $p\in
  \beta(\{0,1\}^*)(\eps\cup\{0,1\})$. 
\end{prop}
\begin{proof}
  Note that $w=\ex(x)$ can be recovered from the prefix $\beta(x)11$:
  the first occurrence of $11$ starting at even position\footnote{We
  enumerate positions in a word starting with~0.} signals that the
  prefix $\beta(x)$ is completed and $x$ is uniquely determined by
  this prefix. 
\end{proof}

Proposition~\ref{common-prefix} will play an important role
below. In particular, we will use the fact that common prefix of words from
$D$ does not contain $0^3$ (easily follows from
Proposition~\ref{common-prefix}). 
Just now we indicate an another simple corollary.

\begin{corollary}\label{antichain}
  No word from $D$ is a prefix of an another word from~$D$.  
\end{corollary}

\begin{rem}
  For an arbitrary language of squares Corollary~\ref{antichain} does
  not hold: $0101$ is a prefix of $010111010111$.
\end{rem}

Now we are ready to prove  a linear bound on the cardinality of a
regular language $L(A)$ that sits in $D$.

\begin{lemma}\label{Dimmune-few}
  Let  $A$ be DFA such that   $L(A)\subset D$. Then $|L(A)|\leq |Q|$, where
  $Q$ is the state set of~$A$.
\end{lemma}
\begin{proof}
  It is sufficient to consider a minimal DFA $B$ recognizing
  $L(A)$. The states of~$B$ are in one-to-one correspondence with
  Myhill -- Nerode classes. Consider two
  words 
  $u_1u_1\ne u_2u_2$ in $L(A)$. We prove that  $u_1u_2\notin
  L(A)$. It implies that  $u_1$ and $u_2$ are not equivalent and the
  number of Myhill -- Nerode classes is not less than $|L(A)|$.

  Suppose that $u_1u_2=vv\in L(A)$. Either $u_1$ is a prefix of~$v$ or
  $v$ is a prefix of~$u_1$. In both cases we come to a contradiction
  with Proposition~\ref{common-prefix}: both $u_1$ and $v$ contain 
  $0^3$ as a subword.
\end{proof}

Note also that lengths of words of a finite regular language $L(A)$ are also
linearly upperbounded by the number of states of DFA $A$.

\begin{prop}\label{FinRegLen}
  Let  $A$ be DFA with the state set~$Q$. If  $|L(A)|<\infty$ then for
  any  $w\in L(A)$
  we have $|w|< |Q|$.
\end{prop}

This fact can be easily extracted from proofs of pumping lemma for
regular languages. An equivalent statement: 
$L(A)$ is infinite iff there exists an accepting walk on the graph of
the automaton~$A$ containing a cycle.

Now we construct algorithms used in the proof of Theorem~\ref{first}.

Note that the binary representation of the length of word~$x$
has a size $O(\log|x|)$. So arithmetic operations with numbers of this
magnitude can be performed using log space. This fact is widely used
below. 

\begin{prop}\label{exNL}
The membership problem for the
language~$D$ is in \DL. 
The map $\ex$ and the inverse map $\ex^{-1}$ are log space computed.
\end{prop}
\begin{proof}
    The algorithm that computes  $\ex$ repeats twice the following procedure:
  compute $\beta(w)$, where $w$ is the input word, add to the result
  $1^2$, then add 0s (the required number of zeroes is computed using
  log space). 

  The membership test for~$D$ can be done in two stages. At first the
  algorithm checks that an input word is a square $uu$ (easily
  performed using log space).  If so, on the second stage the
  algorithm finds an occurrence of $11$ in the word~$u$ starting at
  even position and splits $u$ in a prefix and suffix removing this
  subword $11$. After that it verifies that the prefix is a code
  $\beta(x)$ of some $x$ and the suffix is $0^{|x|^2+3}$.

  Computing of the inverse map is performed in similar way.
\end{proof}

In nondeterministic algorithms working with inputs containing a
description of DFA $A$ we use procedures `guessing a
word' and `guessing a square'.  

Guessing a word is a repeating guessing of symbols from the alphabet
of~$A$. These symbols form a word and in parallel algorithms will test
some properties of this word. Tests should use  log
space and read a word in one-way. Note that DFA description is not
shorter than the alphabet size as well as the cardinality of the state
set. So, log space is sufficient to store a constant number of
symbols and states. 

A simple example of a test is a check that a guessing word is accepted
by~$A$. This test stores a current state of the automaton and applies
the transition function.  It has been used above in the proof of
Theorem~\ref{promise}. 

Guessing a word is also used in algorithmic version of
Proposition~\ref{FinRegLen}. 

\begin{prop}\label{infty}
  Infiniteness of a regular language is in \NL{} provided a language
  is represented by DFA recognizing it.
\end{prop}
\begin{proof}
  The algorithm guesses a state $q$ of an input DFA $A$. After that it
  guesses a word  $w\in L(A)$ such that the automaton visits $q$ at
  least twice while reading~$w$. The latter condition is easily
  verified on log space even if an input should be read in one-way.
\end{proof}

The next example will be used below.

\begin{prop}\label{betaNL}
  Halves of words from $D$, i.e. words in the form
  $\beta(x)1^20^{|x|^2+3}$, can be recognized using log
  space. Moreover, a recognizing algorithm can read input in one way. 
\end{prop}
\begin{proof}
  The algorithm runs in two stages. 

  At the first stage it reads pairs of symbols and counts
  the number of pairs. If the pair
  $00$ is read then the algorithm stops with the negative answer.
  If the pair
  $11$ is read then the algorithm starts the second stage.
  In other cases the counter of pairs is increased by~$1$.

  At the second stage the algorithm computes and stores $x^2+3$, where
  $x$ is the counter value in the end of the first stage. 
  Then the algorithm reads symbols and
  counts the number of symbols. If the symbol $1$ is read then the
  algorithm stops with the negative answer.

  After reading the last symbol the algorithm compares the number of
  0s read and the stored  value
  $x^2+3$. The answer is positive if these values are
  equal. Otherwise, the answer is negative.

  Correctness of the algorithm is clear as well as  logarithmic
  bound on space used.
\end{proof}

Guessing a square works similarly. 
It has two input parameters $q_1,q_2\in Q(A)$. The procedure guesses a
word $w$ such that $\delta_A(s,w)=q_1$,
$\delta_A(q_1,w)=q_2$ (hence $\delta_A(s,ww)=q_2$).  For this purpose
it is sufficient to store two  states---the current states of reading
processes starting  at the state~$s$ and at the state~$q_1$
respectively. In parallel tests for the guessed word might be
launched. They also should use  log
space and read a word in one-way.

Guessing a square is used in the following algorithm that checks
triviality of RR problem for a finite language.  This check will be
applied for RR problems with a  filter containing~$\bar D$. Thus
for instances $A$ such that  $L(A)\sm D\ne\es$ the answer of RR
problem is positive.

\begin{lemma}\label{D-only}
  Testing conjunction $|L(A)|<\infty$ and $L(A)\subset D$ is in  $\NL$.
\end{lemma}

\begin{proof}
  Testing finiteness of a regular language is in \NL{}
  (Proposition~\ref{infty}). So the algorithm in question performs
  this check at the first stage. If $L(A)$ is infinite the algorithm
  outputs the answer and finishes. So in the sequel we assume that
  $L(A)$ is finite.

  Note that $L(A)\setminus
  D\ne\es$ iff  
  \begin{itemize}
  \item either there exists  $w\in L(A)$ of odd length;
  \item either there exists  $w\in L(A)$ of even length that is not a
    square;
  \item or all words in $L(A)$ are squares and there exists $ww\in
  L(A)\setminus D$. 
  \end{itemize}

  The algorithm guesses among these three cases and tests a chosen
  property. 

  Note that by Proposition~\ref{FinRegLen} lengths of words from
  $L(A)$ are not greater than input size (i.e. DFA description). So
  counting up to a word length can be done using log space.

  In the first case the algorithm guesses a word 
  $w\in L(A)$ of odd length (and check the parity of the length). 

  In the second case it  guesses position~$i$,  the length
  $2k$ and word $w$ of the length $2k$ from $L(A)$ such that  $i$th and
  $(k+i)$th bits of $w$ are different (and uses counters to find these
  bits).

  In the last case it guesses a square $ww\in L(A)\setminus D$ and
  checks that $ww\notin D$ using Proposition~\ref{betaNL}. Parameters
  are  states $q$, $t$ such that $\delta_A(s,w) = q$ and
  $\delta_A(q,w)=t$. The algorithm guesses them before guessing a
  square.
 
  Correctness of the algorithm follows from the above observations.
\end{proof}

Our next goal is an algorithm that extracts words accepted by an
automaton recognizing a finite language and outputs their images under
inverse map $\ex^{-1}$. We assume here that the regular language in
question sits in~$D$.

Let $A$ be  DFA such
that $L(A)\subset D$, $Q$ be the state set, $s$ be the initial
state and $Q_a$ be the set of accepting states of~$A$. We denote by $\delta_A$
the transition function $\delta\colon Q\times A^*\to Q$ of~$A$
extended to the set of words in the input alphabet of~$A$ in natural way.

Suppose that all words in $L(A)$ are squares. Define a map
\begin{equation}\label{middle}
  \mu\colon L(A)\to Q
\end{equation}
by the rule: if $w=uu\in L(A)$ then $\mu(w) =
\delta_A(s,u)$. 

Map~\eqref{middle} is not injective. But its preimage is bounded.

\begin{prop}\label{preimage-bound}\label{unique}
  If $L(A)$ is a square language  then for any accepting state~$t$,
  integer~$k$ and a state~$q$ there exists at most one word
  $uu\in\mu^{-1}(q)$ such that $|u|=k
  $ and $\delta_A(s,uu) = t$.
\end{prop}
\begin{proof}
By contradiction. If two different squares $uu$,
$vv$ satisfy these conditions then $uv\in L(A)$ but is not a square.
\end{proof}

\begin{rem}\label{poly-bound}
  Proposition~\ref{preimage-bound} implies a polynomial bound of the
  cardinality of $L(A)$ consisting of squares. 

  From propositions~\ref{FinRegLen} and~\ref{preimage-bound} we get a
   bound
   $$
   |L(A)|=\sum_{q\in Q}|\mu^{-1}(q)|\leq \frac{|Q|}2\cdot |Q|\cdot |Q|
   =|Q|^3/2.
   $$
\end{rem}

\begin{lemma}\label{output}
  There exists an \FNL{} algorithm that outputs the list of words of
  $\ex^{-1}(L(A))$ provided $L(A)\subset D$. 
\end{lemma}

\begin{proof}
  It follows from Proposition~\ref{preimage-bound} that each word
  $uu\in L(A)$ is uniquely determined by states $q\in Q$,
  $t\in Q_a$ and integer~$k$ such that $\mu(uu)=q$, $|u|=k$,
  $\delta_A(s,uu)=t$. Proposition~\ref{FinRegLen} guarantees that
   $k<|Q|$. 

  The algorithm in question tries all possible values $q$, $t$,
  $k$. For each triple it uses an \NL{} oracle to check an existence
  of a word~$uu\in D$ determined by the triple.

  An \NL{} algorithm for this check  guesses a square
  with parameters $q,t$. In parallel it verifies that the length of a
  guessed word~$u$ is~$k$ and $uu\in D$. For the former test it uses a
  counter of guessed symbols. The latter is possible due to
  Proposition~\ref{betaNL}.

  If  $uu\in D$ for a triple $q$, $t$, $k$ does exist then the
  algorithm decodes $x$ such that
  $u=\beta(x)1^20^{|x|^2+3}$ and outputs~$x$.

  It can be done by  guessing a word~$u$
  combined with two simulations
  of  reading
  the word~$u$ by the automaton~$A$. One simulation starts from the initial
  state~$s$ and the another starts from~$q$. These two simulations
  need to store two current states $q_1$ и $q_2$ respectively. 

  Guessing a symbol is replaced in this procedure by two trials.  For
  each possible variant $\al$ of the next symbol (there are two of
  them) the modified algorithm simulates reading the pair $\al\bar\al$
  and for new values $q_1'$ and $q_2'$ asks an \NL{} oracle about
  existence of a word $w$ such that $\delta_A(q_1',w)=q$,
  $\delta_A(q_2',w)=t$. The corresponding \NL{} algorithm is a
  modification of guessing a square. 
  Due to  $\NL=\co\NL$ an negative
  answer can be also verified on nondeterministic log space.

  The word $uu$ is unique for the triple $q$, $t$, $k$. So the oracle
  answers positively for exactly one value of $\al$.  This value is
  the next symbol of $x$ and the algorithm outputs it.
\end{proof}

Tying up loose ends we get the proof of the main result.

\begin{proof}[Proof of Theorem~\ref{first}.]
  Let us prove that 
  \begin{equation}\label{Xex}
        X\lelog \reg(X_\ex)\ledtnl X.
  \end{equation}

  The first reduction does exist due to~\eqref{genred} and
  Proposition~\ref{exNL}. 

  Now we construct the second reduction. Let $A$ be an instance of
  $\reg(X_\ex)$. The reducing algorithm checks infiniteness of $L(A)$
  and $L(A)\setminus D\ne\es$ using Lemma~\ref{D-only}.  If $R$ is
  infinite or it contains a word from $\bar D$ then the reducing
  algorithm forms a query list of  length~1 containing a fixed
  element~$x_0\in X$ and output the list.
 
  Otherwise the algorithm outputs the list of all words from
  $\ex^{-1}(L(A))$ applying the algorithm from Lemma~\ref{output}.

  To prove correctness of the above reduction note that $D$ is
  regularly immune. So any infinite regular language has a common word
  with~$X_\ex$. If $L(A)$ is finite and contains a word from  $\bar
  D$, then
  it has nonempty intersection with $X_\ex$. Finally, if a finite
  language  $L(A)$ is contained in  $D$ then 
  $R\cap X_\ex\ne\es$ iff the output list contains a word from~$X$.
\end{proof}

Taking into account Statement~\ref{HierNormal} we
get the following corollary.

\begin{corollary}\label{ex-complete}
  Each class $\Sigma_k^p$, $\Pi_k^p$ of the polynomial hierarchy
  contains an RR problem that complete for a class under
  $\lenl$-reductions (and under 
  polynomial reductions).
\end{corollary}

\section{Universality of generalized nondeterministic models}\label{GNA-sect}

RR problems are closely related to models of generalized
nondeterminism (GNA) introduced in~\cite{Vya09}. GNA classes are
parametrized by languages of infinite words (certificates). It was
shown in~\cite{Vya11} that each GNA class contains  RR problem
complete for the class under  $\lelog$-reductions. 
A filter of  RR problem is formed by prefixes of certificates for
the GNA class.

Note that filters in the proof of Theorem~\ref{first} are not prefix
closed.  So they do not correspond any GNA class. To
prove universality of GNA classes we need a more sophisticated construction.

What RR problems correspond to GNA classes? To be prefix closed is a
necessary condition only. The second condition is the following: each
filter word is a proper prefix of a filter word. These two conditions
guarantee that the filter is the prefix set for some set of
certificates.

To satisfy the second condition a very simple modification of a filter
is needed.

\begin{prop}
  If $L\subseteq \Sigma^*$, $\#\notin \Sigma$ then
$$
\reg(L)\lelog \reg(L\#^*) \lelog \reg(L).
$$
\end{prop}
\begin{proof}
  The language $L$ is a regular restriction of $L\#^*$.  We get the
  first reduction from this observation and~\cite[Lemma~4]{Vya11}.

  Now we construct the second reduction. Let  $A$ be an instance of the
  problem  $\reg(L\#^*)$. To output an instance $B$ of $\reg(L)$
  the reducing algorithm changes the set of
  accepting states of DFA~$A$  only. A state $q$ is 
  accepting for the automaton~$B$ iff $\delta_A(q,\#^k)$ is an
  accepting state of~$A$ for some~$k$.

  To check correctness of this reduction take a word in the form
  $w\#^k$, $w\in L\cap L(A)$. Then $B$ accepts $w$ by construction. In
  the other direction: if
  $B$ accepts $w\in L$ then after reading $w$  the automaton~$A$ moves to a
  state such that an accepting state is reachable by reading a
  sequence of $\#$. 

  Note that $\reg(\#^*)\in\DL$~\cite[Corollary~1]{Vya11}. Using this
  observation a log space algorithm for the second reduction can be
  easily constructed.
\end{proof}

Thus for any RR problem with a prefix closed filter there exists an
equivalent RR problem with a filter that coincides with the prefix set
of a language of infinite words.

So universality of GNA classes follows from the following
generalization of Theorem~\ref{first}.

\begin{theorem}\label{third}
  For any nonempty language $X$ there exists a prefix closed filter
  $P$ such that
  \begin{equation}\label{prefred}
    X\lelog \reg(P)\ledtnl X.
  \end{equation}
\end{theorem}

To prove Theorem~\ref{third} we follow the same plan as for
Theorem~\ref{first}. 

Again we use monoreductions. We choose the
filter $P$ in the form $X'_f$ for some map~$f$ and some set $X'$ which
is $\lelog$-equivalent to~$X$. Actually  we will use $X'=0X\cup\{\eps\}$. 

To make the filter~$P$ prefix closed we choose an appropriate
map~$f$. Namely, we require that $f$ maps words in the form 
$0\{0,1\}^*$ to a prefix code~$D$ (i.e. a set of words such that no
word in the set is a prefix of an another word in the set).
Also we require that~$f$  maps all nonempty words to the set $\D= \{w:
w= uv, u\in D\}$ of all possible extensions of~$D$ words. Then
$f(\overline{X'})$ is suffix closed while $P=X'_f =
\overline{f(\overline{X'})}$ is prefix closed.

By construction $P$ contains $\bar\D$. So for an instance $A$ of the
RR problem 
$\reg(P)$ the answer is positive if
$L(A)\sm\D\ne\es$. The latter condition plays now a role of
infiniteness condition in the proof of Theorem~\ref{first}. 

To keep arguments from the previous section a check of
$L(A)\sm\D\ne\es$ should belong to \NL{} and generation of the list
containing all prefixes of words from $L(A)\cap D$ should be done by a
$\FNL$ algorithm provided $L(A)\subset\D$. 

Now we elaborate details of this plan.

We use the set $D$ from the previous section as a prefix code (see
Corollary~\ref{antichain}). 
  
The set  $\D$ is not regularly immune. But in the crucial case
$L(A)\subset \D$ (see the above plan) the $D$-prefix set  
\begin{equation*}
R_D=\{ w: w\in D\ \text{and}\ wv\in L(A)\}
\end{equation*}
is finite. Below we modify algorithms from Section~\ref{red-sect} to
deal with the sets of this form.

Now we upperbound the cardinality of $R_D$ and lengths of words in $R_D$.

\begin{prop}\label{short}
  If $L(A)\subset\D$, where $A$ is DFA with the state set~$Q$,
  $|Q|=n$, and $w\in R_D$ then
  $$|w|\leq 2(n-3)+4\sqrt{n-3}+10.$$
\end{prop}
\begin{proof}
  By definitions a word $w\in D$ has a form
  $$\beta(x)1^20^{\ell^2+3}\beta(x)1^20^{\ell^2+3},$$ where
  $\ell=|x|$ and  $wz\in L(A)$ for some~$z$. 

  If $\ell^2+3\geq n$ then while reading the word $0^{\ell^2+3}$ the
  automaton visits some state at least two times. Thus  for
  some integer~$k>0$ and all integers $i\geq0$ we have
  $$
  w_i = \beta(x)1^20^{\ell^2+3+ik}\beta(x)1^20^{\ell^2+3}z\in L(A).
  $$
  Note that for  $i>0$ a word $w_i$ does not have a $D$ prefix in
  contradiction with the hypothesis $L(A)\subset\D$. Indeed such a
  prefix must have a prefix $\beta(x)1^2$ and due to
  Proposition~\ref{common-prefix} must coincide with the word~$w$. But
  $w$ is not a prefix of $w_i$ for $i>0$.

  Thus $\ell^2+3<n$. The length of $w$ equals 
  $2\ell^2+4\ell+10$. It implies the required inequality.
\end{proof}

To upperbound the cardinality of $R_D$ we modify
Proposition~\ref{preimage-bound}.   

\begin{prop}\label{unique2}
  Let $L(A)\subset \D$. For any coaccessible\footnote{A coaccessible
  state is a state from which it exists a path to an
  accepting state.}  state~$t$, integer~$k$ and a
  state~$m$ there exists at most one word $uu\in D$ such that $|u|=k
  $ and $\delta_A(s,u) = m$, $\delta_A(m,u)=t$.
\end{prop}
\begin{proof}
  Suppose there are two words $uu$ and $vv$ satisfying these 
  conditions. Then $uvw\in L(A)$ for some~$w$. From
  Proposition~\ref{common-prefix} we conclude that a word $uvw$ does
  not contain $D$ prefixes in contradiction with the hypothesis
  $L(A)\subset \D$. 
\end{proof}

\begin{corollary}\label{squares-few}
  If $L(A)\subset \D$ then $|R_D|\leq2|Q|^3(1+o(|Q|))$.
\end{corollary}

Correctness of the algorithm from Lemma~\ref{output} is based on
Proposition~\ref{preimage-bound}. Replacing it by
Proposition~\ref{unique2} we get the following lemma. 

\begin{lemma}\label{list}
  There exists a map $f\in\FNL$ that generates the element list of
  $\ex^{-1}(R_D)$ provided $L(A)\subset\D$.
\end{lemma}
\begin{proof}
  The algorithm mimics the algorithm from Lemma~\ref{output} but for
  guessed squares it tests a membership to prefixes of~$L(A)$ (instead
  of a membership to~$L(A)$ itself). The prefix set of $L(A)$ is
  recognized by DFA $A_p$ having the same state set and the transition
  table as DFA~$A$. Accepting states of $A_p$ are coaccessible states
  of~$A$.

  It is possible to check coaccessibility on nondeterministic log
  space. So a \FNL{} algorithm is able to apply this test.

  Correctness of the algorithm follows from Proposition~\ref{unique2}. 
\end{proof}

\begin{lemma}\label{upD-only}
  Testing $L(A)\subset\D$ is in \NL. 
\end{lemma}

\begin{proof}
  The statement is equivalent to
  $\reg(\overline{\D})\in \NL$ due to equality $\NL=\co\NL$. 

  The algorithm is similar to the algorithm from Lemma~\ref{D-only}.

  Let DFA  $A$ with the state set $Q$ be an instance of
  $\reg(\overline{\D})$.
  The algorithm guesses a word from $\overline\D\cap L(A)$ and verifies this
  condition. It uses log space and reads an input in one
  way.

  To detect a word from $\overline\D\cap L(A)$ the algorithm guesses a
  prefix of a word~$w\in L(A)$ certifying that $w\notin\D$.

  At first a prefix in the form
  \begin{equation}\label{var-first}
    \beta(x)(\eps\cup\{0,1\}),\quad |x|=5|Q|,     
  \end{equation}
  guarantees that $L(A)\sm\D\ne \es$ due to Proposition~\ref{short}.

  Note that each language from the following list
  \begin{align}
    &  \beta(x)00, && \\
    &  \beta(x)110^i1, &&i<|x|^2+3,\\
    &  \beta(x)110^i, &&i>|x|^2+3,\\
    &  \beta(x)110^{|x|^2+3}\beta(y)00, &&\\
    &  \beta(x)110^{|x|^2+3}\beta(y), &&|x|<|y|,\\
    &  \beta(x)110^{|x|^2+3}\beta(y)11, &&|x|>|y|,\\
    &  \beta(x)110^{|x|^2+3}\beta(y)110^i1, &&i<|y|^2+3, && |x|=|y|,\\
    &  \beta(x)110^{|x|^2+3}\beta(y)110^i, &&i>|y|^2+3, && |x|=|y|.
    \label{var-last}
  \end{align}
  do not contain neither  words having a prefix from~$D$ nor prefixes of
  words from~$D$. 

  Also, each
  condition~(\ref{var-first}--\ref{var-last}) is recognized by a log  space
  algorithm reading an input in one way provided the length of a word
  is less then $10|Q(A)|+1$.

  With the same restrictions it is possible to verify that the
  language $L(A)$ contains a 
  word $w$ such that $w$ is a prefix of a word from a language
  presented in the list~(\ref{var-first}--\ref{var-last}). More
  exactly, for~\eqref{var-last} we exclude prefixes in the form
  \begin{equation}\label{eq-pref}
    \beta(x)110^{|x|^2+3}\beta(y)110^{|x|^2+3},\quad |x|=|y|.
  \end{equation}

  If  $L(A)$ do not contain any word satisfying the above requirements
  then each word from $L(A)$ has a prefix in the
  form~\eqref{eq-pref}. 
  In this case the algorithm should indicate a prefix with 
  $x\ne y$. If such a prefix does not exist we have $L(A)\subset \D$.

  To indicate a prefix in the form~\eqref{eq-pref} with $x\ne y$ the
  algorithm guesses two states $q,t$ such that $t$ is coaccessible
  (and checks coaccessibility by consulting with \NL{} oracle).

  The algorithm also guesses states 
  \begin{equation}\label{q1q2}
    q_1=\delta_A(s,\beta(p)),\quad
    q_2=\delta_A(q,\beta(p)),
  \end{equation}
  where $p$ is the greatest common prefix of $x$ and $y$. Then it
  simulates reading $\beta(p)$ starting from states $s$ and
  $q$ keeping the number of guessed symbols in a counter. It continues
  by reading $01$ in one case and by reading $10$ in the another (to 
  guarante that $x\ne y$). After that the algorithm guesses two
  words separately verifying for both of them the
  format~\eqref{eq-pref} and counting their lengths. When guessing is
  finished the algorithm checks that the above conditions~\eqref{q1q2}
  and compares the lengths of words.
\end{proof}

Now we present a map~$\dex$ to construct the second reduction
in~\eqref{prefred}: 
\begin{equation}\label{dex-def}
  \begin{aligned}
    \dex\colon \eps&\mapsto\eps,\\
    \dex\colon 0x&\mapsto \ex(x),\\
    \dex\colon 1x&\mapsto \ex(\pi_1(x))\pi_2(x),
  \end{aligned}
\end{equation}
where $\pi(x) = (\pi_1(x),\pi_2(x))$ is a bijection $\{0,1\}^*$ onto
$\{0,1\}^*\times(\{0,1\}^*\setminus\{\eps\})$ such that $\pi$,
$\pi^{-1}$ are  log space computable.

It immediately follows from definition that  $\dex$  is injective and
$\Im(\dex)=\D\cup\{\eps\}$. It easy to see that $\dex$ and $\dex^{-1}$
are log space computable.

Note that a bijection $\pi$  in~\eqref{dex-def} can be easily
constructed from a map
\begin{equation}\label{phi-def}
  \ph\colon (x,\al y) \mapsto \beta(x)\al\al y,
\end{equation}
which sends $\{0,1\}^*\times(\{0,1\}^*\setminus\{\eps\})$ to a set $S$
consisting of words of the length greater than~3 and words $00$, $11$.
Namely,  $\pi(x) =
\ph^{-1}(\iota(x))$, where $\iota$ is an injective map to~$S$ such
that both $\iota$ and $\iota^{-1}$ are log space computable.

For the sake of completeness we presented a  construction of such a map.

\begin{lemma}\label{finite-exception}
  Let $S$ be a cofinite set of binary words, i.e. $|\bar S|<\infty$.
  Then there exists an injective map $f_S$  such that domain of $f_S$
  is the set of all words, $\Im(f_S)=S$ and 
  $f_S$, $f_S^{-1}$ are log space computable.
\end{lemma}
\begin{proof}
  Choose a set of binary words $M$ such that all words in $M$ have the
  same length~$n$,  $|M|=|\bar S|$ and
  $$
  n >\max_{x\in\bar S}|x|.
  $$
  Then choose a bijection $\vk\colon \bar S\to M$. Now a map $f_S$ in
  question can be defined as follows
  \begin{equation*}\label{fS-def}
    \begin{aligned}
      f_S\colon x&\mapsto 1\vk(x), && \text{for } x\in \bar S,\\
      f_S\colon 0^i1m&\mapsto 0^{i+1}1m,&& \text{for } q\in M,\\
      f_S\colon x&\mapsto x && \text{otherwise.}
    \end{aligned}  
  \end{equation*}
  It is obvious that $f_S$ satisfies all requirements.
\end{proof}

\begin{proof}[Proof of Theorem~\ref{third}]
  It is observed above that the language $X$ is equivalent to the language
  $X'=0X\cup\{\eps\}$ under 
  $\lelog$-reductions. From definition~\eqref{dex-def} the map $\dex$
  sends nonempty words from $X'$ to the set $D$ while the images of
  all nonempty words coincide with~$\D$. It means that
  $\dex(\overline{X'})$ is suffix closed, i.e. $w\in
  \dex(\overline{X'})$ implies $wv\in \dex(\overline{X'})$ for
  all~$v$, and $X'_\dex =
  \overline{\dex(\overline{X'})}$ is prefix closed.

  Therefore it is sufficient to show that
  $$
  \reg(X'_\dex)\ledtnl X'.
  $$
  The reduction is constructed using Lemmata~\ref{list}
  and~\ref{upD-only} in a way similar to the  proof of
  Theorem~\ref{first}.  

  Using Lemma~\ref{upD-only}
  the reducing algorithm checks $L(A)\subset\D$ for an automaton $A$
  which is an instance of the problem $\reg(X'_\dex)$. If
  $L(A)\not\subset\D$ then the answer for this instance is positive
  and the reducing algorithm outputs a trivial query list containing
  the fixed word from~$X'_\dex$. Otherwise the reducing algorithm
  invokes the procedure from Lemma~\ref{list} to generate the element list of
  $\ex^{-1}(R_D)$.
  It follows from construction of the set $X'_\dex$ that $L(A)\cap
  X'_\dex\ne\es$ iff $x_j\in X'$ for some $j$.
\end{proof}

Taking into account normality of classes in the polynomial hierarchy
we get the following corollary. 

\begin{corollary}
  For each $k$ the classes $\Sigma^p_k$, $\Pi^p_k$ contain RR problems
  with a prefix-closed filter
  that are
  complete for a class under polynomial reductions. 
\end{corollary}

\end{document}